\documentclass[12pt]{article}
\usepackage[margin=2cm]{geometry}
\usepackage{amsmath,amstext,amsfonts,amsthm,amssymb,bbm,hyperref}
 \usepackage{romanbar}
\usepackage[colorinlistoftodos]{todonotes}
\usepackage{wrapfig}

\theoremstyle{plain}
\newtheorem{thm}{Theorem}
\newtheorem{defin}{Definition}[section]

\newtheorem{prop}[defin]{Proposition}

\newtheorem{rmk}[defin]{Remark}
\newtheorem{cor}[defin]{Corollary}
\newtheorem{lemma}[defin]{Lemma}
\newtheorem*{thm*}{Theorem}

\def\Sym{\mathfrak S}

\def\H{\mathfrak H}
\def\N{\mathcal N}
\def\Zd{\mathbb Z^d}
\def\Even{\mathbb{Z}^{d,\operatorname{even}}}
\def\Tab{\operatorname{Tab}}
\def\Bitab{\operatorname{Bitab}}

\def\Pl{\operatorname{Pl}}
\def\head{\operatorname{head}}
\def\tail{\operatorname{tail}}
\def\p{{\mathbf p}}
\def\lam{{\boldsymbol\lambda}}
\def\t{{\boldsymbol t}}
\def\s{{\boldsymbol s}}
\def\r{{\boldsymbol r}}

\begin{document}
\title{The umpteen operator and its Lifshitz tails}
\author{Ohad N.\  Feldheim$^1$  and Sasha Sodin$^2$}
\maketitle

\footnotetext[1]{Einstein Institute of Mathematics, The Hebrew University of Jerusalem, Givat Ram, Jerusalem, Israel. Email: ohad.feldheim@mail.huji.ac.il. Supported in part by ISF grant 1327/19.} 
\footnotetext[2]{School of Mathematical Sciences, Queen Mary University of London, London E1 4NS, United Kingdom. Email: a.sodin@qmul.ac.uk. Supported in part by the European Research Council starting grant 639305 (SPECTRUM) and by a Royal Society Wolfson Research Merit Award.}

\begin{abstract}
As put forth by Kerov in the early 1990s and elucidated in  subsequent works, numerous properties of Wigner random matrices are shared by certain linear maps playing an important r\^ole in the representation theory of the symmetric group. We introduce and study an operator of representation-theoretic origin which bears some similarity to discrete random Schr\"odinger operators acting on the $d$-dimensional lattice. In particular, we define its integrated density of states and prove that in dimension $d \geq 2$ it boasts Lifshitz tails similar to those of the Anderson model.

The construction is closely related to an infinite-board version of the fifteen puzzle, a popular sliding puzzle
from the XIX-th century. We estimate, using a new Peierls argument, the probability that  the puzzle returns to its initial state after $n$ 
random moves.
The Lifshitz tail is deduced using an identification of our random operator with the action of the adjacency
matrix of the puzzle on a randomly chosen representation of the infinite symmetric group.
\end{abstract}

\section{Introduction} 

\paragraph{A family of operators} Let $G = (V, E)$ be a finite or countable connected graph of bounded degree. Denote by $\Sym[V]$ the group of finitely supported permutations of $V$, and let
\[\begin{split}
 \ell_2(\Sym[V]) &= \left\{ u = \sum_{\pi \in \Sym[V]} c_\pi \, \pi \, \middle| \, \|u\|^2
 \overset{\text{def}}{=} \sum_{\pi \in \Sym[V]} |c_\pi|^2 < \infty\right\}~,\\
\H [V] &= \ell_2(V \to \ell_2(\Sym[V])) = \left\{ \psi: V \to \ell_2(\Sym[V]) \, \middle| \, \|\psi\|^2  \overset{\text{def}}{=} \sum_{x \in V} \| \psi(x)\|^2 < \infty \right\}~. 
\end{split}\]
The subject of this note is the bounded self-adjoint operator $H[G]$ acting on $\H[V]$ via
\begin{equation}\label{eq:defH} (H[G]\psi)(x) = \sum_{y \sim x} (x\, y) \psi(y)~,\end{equation}
where the sum is taken over the vertices $y \in V$ which are adjacent to $x$, and $(x\, y) \in \Sym[V]$ is the transposition exchanging $x$ and $y$. 
We shall freely switch between (\ref{eq:defH}) and the block matrix representation $H[G] = (H[G](x,y))_{x,y \in \Zd}$, where the blocks
\begin{equation}
H[G](x, y) = \begin{cases}
(x\, y)~, &x \sim y \\
0
\end{cases} 
\end{equation}
represent operators acting on $\ell_2(\Sym[V])$. 

The motivation to study $H[G]$ comes from its interpretation as a random operator. If $G$ is finite, this interpretation is based on the classical representation theory of the symmetric group. For each irreducible representation $\lam: \Sym[V] \to E_\lam$ (where $E_\lam$ is the ambient space), let $H[G; \lam]$ be the operator acting on $\ell_2(V \to E_\lam)$ via
\[ (H[G; \lam] \psi)(x) = \sum_{y \sim x} \lam \big( (x\, y) \big)  \psi(y)~. \]
Then 
\begin{equation}\label{eq:domp-h-finite} H[G] \simeq \bigoplus_{\lam \in \operatorname{Irrep}\Sym[V]} H[G; \lam] \oplus \cdots \oplus H[G; \lam]~,\end{equation}
where the addend corresponding to an irreducible representation $\lam \in \operatorname{Irrep}\Sym[V]$ appears $\dim \lam = \dim E_\lam$ times. Define the Plancherel probability distribution $\Pl_N(\{\lam\}) = \frac{\dim^2 \lam}{|V|!}$ on $\operatorname{Irrep} \Sym[V]$, then (\ref{eq:domp-h-finite}) implies that the normalised eigenvalue counting function  
\[ \N_G(\lambda) = \frac{1}{|V| \times |V|!} \times \text{number of eigenvalues of $H[G]$ in $(-\infty, \lambda]$} \]
is equal to the expectation $\N_{G}(\lambda) = \mathbb E \N_{G; \lam}(\lambda)$ of the normalised eigenvalue counting function
\[ \N_{G; \lam}(\lambda) = \frac{1}{|V| \times \dim \lam} \times \text{number of eigenvalues of $H[G; \lam]$ in $(-\infty, \lambda]$}~,\]
corresponding to a representation $\lam$ chosen at random according to $\Pl_N$ (here and forth eigenvalues are counted with multiplicity).

The properties of the random cumulative distribution function $\N_{G; \lam}(\lambda)$ and of the average $\N_G(\lambda)$ can be described in detail for the case of the complete graph $H = K_N$. As $N \to \infty$, one has the following semicircular asymptotics, which we have learnt from Alexey Bufetov: if $\lam$ is chosen at random according to $\Pl_N$, then
\begin{align}
\label{eq:meanfield''}
 \mathbb E \N_{K_N; \lam} (\sqrt N \lambda) \longrightarrow \int_{-\infty}^\lambda \frac1{2\pi } \left[ \max(0, 4-s^2) \right]^{\frac 12} ds~,
\\
\label{eq:meanfield'}
 \N_{K_N; \lam} (\sqrt N \lambda) \overset{\operatorname{distr}}{\longrightarrow}  \int_{-\infty}^\lambda \frac1{2\pi } \left[ \max(0, 4-s^2) \right]^{\frac 12} ds~.
\end{align}
These relations illustrate one of the numerous common properties between  $H[K_N]$ and Wigner random matrices. The parallelism between the representation theory of large symmetric groups and random matrix theory was first put forth by Vershik \cite{Ver}, which triggered developments in many different directions. Here we focus on one particular connection, found by Kerov: loosely speaking, $H[K_N]$ is a representation theoretic counterpart of the Gaussian Unitary Ensemble (or, more precisely, of its direct integral over the realisations of the randomness). In particular, the relations (\ref{eq:meanfield''})--(\ref{eq:meanfield'})  can be proved using the method of moments, similarly to Wigner's proof of the semicircular law for random matrices \cite{Wigner}: letting $\tau: \ell_2(\Sym[V]) \to \mathbb C$ be the functional sending $\sum c_\pi \pi$ to $c_{\mathbbm 1}$, the coefficient of the identity, we have
\[\begin{split} \lim_{N\to\infty}\frac{1}{N \times N!} \operatorname{tr} (H[K_N]/\sqrt N)^n 
&= \lim_{N\to\infty}\frac1N \sum_{j=1}^N \tau \left( (H[K_N]/\sqrt N)^n (j, j)\right) \\
&= \begin{cases}
\frac{n!}{(n/2)!((n/2)+1)!}~, & \text{$n$ is even}\\
0
\end{cases} 
= \frac{1}{2\pi} \int_{-\infty}^\infty  \lambda^n \left[ \max(0, 4-\lambda^2) \right]^{\frac 12} d\lambda~,\end{split}\]
which implies (\ref{eq:meanfield''}); a similar computation of the variance implies (\ref{eq:meanfield'}). We omit the details (see e.g.\ \cite{me-Ker} for  computations of this kind).

On the other hand, $H[K_N]$  enjoys numerous symmetries, for exapmple, the trace of any polynomial of $H[K_N]$ lies in the center of the group algebra (as emphasised by Bufetov, this and other properties of $H[K_N]$ are parallel to those of the Perelomov--Popov matrices appearing in the representation theory of classical Lie groups; see \cite{PP} and the recent works of Bufetov--Gorin \cite{BufGor} and Collins--Novak--\'Sniady \cite{CNS}). Consequently, $H[K_N]$ can be explicitly diagonalised. Using the block matrix representation of Jucys--Murphy elements due to Biane \cite[Proposition 3.3]{Biane} (a closely related result was obtained earlier by Gould \cite{Gould}), one can show that (\ref{eq:meanfield'}) is equivalent to Kerov's semicircular law for the transition measure associated to a Young diagram chosen at random according to the Plancherel measure \cite{Kerov, Kerov-book}; as shown by Kerov, the latter is equivalent to  the Logan--Shepp--Vershik--Kerov theorem \cite{LS,VK1,VK2} on the limit shape of random Young diagrams. We refer to \cite{Kerov-book} and also to \cite{me-Ker} and references therein for further discussion, and to the works \cite{Biane,Ok,Ok2,Joh,BG} and references therein for some of the additional connections between random matrix theory and the representation theory of the symmetric group.

\paragraph{The main result}
Here we explore the properties of $H[\Zd]$, which is, very loosely speaking, a representation theoretic counterpart of a random Schr\"odinger operator. It has much less symmetries than the mean-field operator $H[K_N]$, and is therefore probably impossible to diagonalise explicitly. On the other hand, it is sensitive to the geometry of the underlying lattice $\Zd$, and as such exhibits interesting features reminiscent of those known  in the theory of disordered systems. 

Our main result pertains to one of the basic objects of study, the integrated density of states, which we now define.  Let $B_{L,d} = [-L, L]^d \subset \Zd$. 
\begin{lemma}\label{l:ids} Let $d \geq 1$. As $L \to \infty$, the sequence of probability measures with cumulative distribution function $\N_{B_{L, d}}(\lambda)$ converges weakly to a probability measure with support equal to $[-2d, 2d]$, the cumulative distribution function  $\N_{\Zd}(\lambda)$  of which is uniquely characterised by the relations 
\begin{equation}\label{eq:ids}
\forall p \in \mathbb C[\lambda] \quad \int p(\lambda) d\N_{\Zd}(\lambda) = \tau\big(p(H[\Zd])(0, 0)\big)~.
\end{equation}
\end{lemma}

\noindent In other words, $\N_{\Zd}$ is  equal to the spectral measure of $H[\Zd]$ corresponding to the vector $\delta_0 \mathbbm 1$,
\[ (\delta_0 \mathbbm 1)(x) = \begin{cases}
\mathbbm 1~, &x = 0\\
0~.
\end{cases}\]
The proof of Lemma~\ref{l:ids}, mostly mimicking standard arguments from the theory of random operators, is reproduced in Section~\ref{s:ids}. 

The limiting function $\N_{\Zd}(\lambda)$ is called the integrated density of states associated with $H[\Zd]$.  The nomenclature is motivated by a interpretation of $H[\Zd]$ as a random operator (or more precisely, as a direct integral over the randomness  of a family of operators $H[\Zd; \t]$ depending on a random parameter $\t$ taking values in the space of Young bitableaux) which we describe,  building on  the work of Vershik and Tsilevich \cite{VT}, in Section  \ref{s:rand2}. We shall prove (see Corollary~\ref{cor:vt})  that $\N_{\Zd}$ is equal to the expectation
\begin{equation}\label{eq:decomp-vt}\N_{\Zd}(\lambda) = \mathbb E \, \N_{\Zd, \t} (\lambda)~, \end{equation}
of a random cumulative distribution function $\N_{\Zd, \t} (\lambda)$. This property of $\N_{\Zd}$, as well as the one  stated in Lemma~\ref{l:ids}, is suggestively similar to the properties of the integrated density of states of metrically transitive operators on $\Zd$ as described, for example, in the monograph of Pastur and Figotin \cite{PF}.

\medskip\noindent
In dimension $d = 1$,  the integrated density of states can be explicitly computed. Indeed,
\[ \tau(H^{n}(0, 0)) = \begin{cases}
\binom{n}{n/2}~, &\text{$n$ is even} \\0
\end{cases} \]
is the number of paths of length $n$ starting and terminating at the origin, which can be shown to imply that
\begin{equation}\label{eq:1d} \N_{\mathbb Z}(\lambda) = \frac{1}{\pi} \arcsin \sqrt{(2+ x)/4}~;\end{equation}
in particular, $\N_{\mathbb Z}$ has a square root singularity at the edges $\lambda = \pm 2$. Our main result is that for $d \geq 2$ the integrated density of states exhibits the following asymptotics, known in the context of random operators as (quantum) Lifshitz tails:
\begin{thm}\label{thm} For each $d \geq 2$, there exist $C>0$ and $c>0$ such  that
\begin{equation}\label{eq:main} c \exp \left\{-C\epsilon^{-\frac d2} \log  (\frac{1}{\epsilon} + 2)\right\} \leq \N_{\Zd} (-2d+ \epsilon) \leq C \exp \left\{-c\epsilon^{-\frac d2}\right\}~, \quad 0 < \epsilon \leq 1~.\end{equation}
\end{thm}
\begin{rmk}
Due to the bipartite structure of $\Zd$, $\N_{\Zd}(-\lambda) = 1 - \N_{\Zd}(\lambda)$, hence  the right tail $1 - \N_{\Zd}(2d - \lambda)$ exhibits the same asymptotics.
\end{rmk}

We recall that, in the context of random operators, Lifshitz tails were introduced by I.~M.~Lifshitz \cite{Lif}. The exponent $\frac d 2$ in (\ref{eq:main}) is characteristic of the so-called quantum fluctuating boundaries (this terminology is explained in \cite{PF}). In the original setting of Lifshitz, \cite{Lif} predicted the logarithmic asymptotics 
\begin{equation}\label{eq:lif-str}
-  \epsilon^{d/2}  \log\N(\lambda_{\min} + \epsilon)  \to c_* \in (0, \infty)~, \quad \epsilon \to + 0~.
\end{equation}
For other models, (\ref{eq:lif-str}) requires logarithmic corrections.

The first mathematical proofs  of Lifshitz tails (for models in the continuum) were obtained by Pastur \cite{Pastur2, Pastur3, Pastur1}, who analysed the Feynman--Kac representation of the semigroup generated by $H$ with the help of  the large deviation estimates of Donsker and Varadhan \cite{DV1}; see further  \cite{PF} and references therein. The counterpart of this method for random operators on the lattice was developed by Biskup and K\"onig \cite{BK}, as part of their work on the parabolic Anderson model. These works provide logarithmic asymptotics such as (\ref{eq:lif-str}).

An alternative approach to Lifshitz tails, going back to the work of Kirsch--Martinelli \cite{KM} and  Simon \cite{Simon}, is based on Dirichlet--Neumann bracketing, i.e.\ bounding the operator from above and below by a direct sum of its finite-volume restrictions with properly adjusted boundery conditions. This approach is technically simpler and more robust but usually leads to less precise (doubly logarithmic) asymptotics. In our setting, a bracketing argument easily leads to a lower bound such as in (\ref{eq:main}) (perhaps, with a higher power of the logarithm); on the other hand, we have not been able to use it to prove the upper bound.  Yet another approach to Lifshitz tails (in the continuum) was recently introduced by David, Filoche, and Mayboroda \cite{DFM}; it relies on the analysis of the so-called landscape function, introduced and studied by the same authors.

Our  proof of (\ref{eq:main}) is based on the analysis of the moments of the operator (see (\ref{eq:propmain})), and is thus morally closer to the Feynman--Kac  approach. We do not rely on precise large deviation estimates (as these are not available in our non-commutative setting); instead, we use a simple variant of the Donsker--Varadhan estimates (Lemma~\ref{l:1}) and combine it with a new and relatively robust Peierls-type argument. This method may be of independent interest even in the classical setting of random Schr\"odinger operators, for models in which bracketing is unavailable or hard to implement.  

To conclude this brief survey, we mention the work of Bapst and Semerjian \cite{BS}, who applied the moment method to the study of Lifshitz tails for random Schr\"odinger operators on a tree, where bracketing also runs into difficulties. Note however that the proof of the upper bound in \cite{BS} relies on an unproved hypothesis, and the full mathematical proof of the Lifshitz tails on the Bethe lattice was accomplished by Hoecker-Escuti and Schumacher \cite{HES} by different methods.

\paragraph{The fifteen puzzle} 

\begin{wrapfigure}{r}{0.25\textwidth}
\includegraphics[scale=.7]{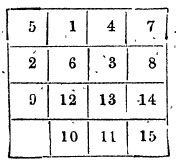}
\caption{An illustration to the fifteen  puzzle from \cite{Ha-Tzefirah} (1880)}\label{fig:15}
\end{wrapfigure}
The proof of our main result is based on an analysis of a probabilistic problem pertaining to a generalised version of the classical fifteen puzzle, the definition of which we now recall. Given a graph $G = (V,E)$, a state of the fifteen puzzle $\Romanbar{15}[G]$ on $G$ is a pair $(\pi, x) \in \Sym[V] \times V$. Two states $(\pi_1, x_1)$ and $(\pi_2, x_2)$ are called adjacent if $x_2 \sim x_1$ and $\pi_2 = (x_2\, x_1) \pi_1$.

Less formally, a state $(\pi, x^*)$ is composed of a marked vertex $x^*$ and a label $\pi(x)$ placed on each vertex $x$. A legal move takes $(\pi, x^*)$ to an adjacent state by exchanging the labels of $x^*$ and one of its neighbours, which now takes the r\^ole of $x^*$. The classical fifteen puzzle, popular since the late 1870-s, is recovered by taking $G = \{1, 2, 3, 4\}^2$ with the graph structure inherited from $\mathbb Z^2$. The marked vertex $x^*$ corresponds to the empty square (see Figure~\ref{fig:15}).

As we emphasise below, the operator $H[G]$ is essentially the adjacency matrix of $\Romanbar{15}[G]$. Using this connection, we reduce the proof of the main theorem  to estimating the probability $p_{2n}(\Zd)$ that the puzzle on the infinite board $\Zd$ returns to the original state after $2n$ independent random moves. Our main auxiliary result provides a  two-sided bound on $p_{2n}(\Zd)$:
\begin{thm}\label{thm:prob} For any $d \geq 2$, there exist $C, c > 0$ such that for any $n \geq 0$
\begin{equation}\label{eq:propmain-intr} c \exp\left\{-C n^{\frac{d}{d+2}} \log^{\frac{2}{d+2}}(n+2)\right\} \leq  p_{2n}(\Zd) \leq C \exp\left\{-c n^{\frac{d}{d+2}} \right\}~.\end{equation}
\end{thm}

The fifteen puzzle has been the subject of numerous mathematical studies, some of which we now mention. 
The connected components of the fifteen puzzle on finite graphs were fully described by Wilson \cite{Wilson}, following extensive earlier research pertaining to special families of graphs such as  $[1, \cdots, L]^2$. 
The asymptotic and probabilistic aspects of the problem were put forth by Diaconis \cite{Dia}, who  asked what is the mixing time of the random walk on a connected component of the state space for $G = [1, \cdots, L]^2$ and $G = (\mathbb Z / L \mathbb Z)^2$. The latter question was answered by Morris and Raymer \cite{MR}, who showed that the mixing time is of order $L^4 \log L$:
\begin{equation}\label{eq:mr} c L^4 \log L \leq T_{\operatorname{mix}} \left( \Romanbar{15}[(\mathbb Z / L \mathbb Z)^2] \right)\leq C L^4 \log L~.\end{equation} 
In contrast to the case of conventional random Schr\"odinger operators, the estimates (\ref{eq:propmain-intr}) and  (\ref{eq:mr}) do not seem to be comparable.

Finally, we remark that the relation between the integrated density of states $\N_{\Zd}$ and the random walk on $\Romanbar{15}[\Zd]$ (a Cayley graph of the symmetric group $\mathfrak S[\Zd]$) is roughly as that  between the integrated density of states of usual Schr\"odinger operators with Bernoulli potential and the random walk on the wreath product of $\{-1, +1\}$ and $\Zd$ (viz., the lamplighter group).

\paragraph{Open questions}

Bridging the logarithmic gap between the upper and lower bounds in Theorem~\ref{thm:prob} would result in matching bounds also in Theorem~\ref{thm}. We believe that the lower bound in both theorems is sharp, up to the values of the numerical constants.

The analogy with random operators naturally leads to  numerous follow-up questions. Do the operators $H[\Zd]$ boast any counterpart of Anderson localisation (see the monographs \cite{PF,AW} and references therein)? One can ask, for example, what is the  type of the measures corresponding to $\N_{\Zd, \t}(\lambda)$ in (\ref{eq:decomp-vt}), particularly, could these measures have pure point components near the edges. A related question is whether the quantum dynamics generated by the direct addends of $H[\Zd]$ is localised near the spectral edges. 

As hinted by (\ref{eq:1d}), the spectral properties of $H[\Zd]$ are trivial for $d=1$; still, they are most probably non-trivial already when $\mathbb Z$ is replaced by the strip $\mathbb Z \times \{0, 1\}$ of width two, for which  the transfer matrix method could perhaps be available.

One should bear in mind that, unlike Anderson-type operators, our quantum particle is coupled to the environment, therefore the analogy with the Anderson model should not be taken too literally. 
We believe that this circle of questions merits further investigation.

Finally, it would also be interesting to find other classes of groups for which once could construct operators similar to $H[\Zd]$. Particularly, we wonder whether a result similar to Theorem~\ref{thm} holds for an operator corresonding to the infinite unitary group. The analogy with the Perelomov--Popov matrices (see \cite{BufGor}) suggests that such an operator should act on functions from $\mathbb Z^d$ to the universal enveloping algebra of the unitary group, and have the block structure
\[ H^{\mathcal U}[\Zd](x, y) = \begin{cases}
E_{xy}~, &x \sim y \\
0~,
\end{cases}\]
where $E_{xy}$ are the generators of the universal enveloping algebra.

\section{Apology for the term ``integrated density of states''}

The goal of this section is to convince the reader that the integrated density of states as defined in (\ref{eq:ids}) shares some properties with the integrated density of states of conventional random operators. In the first part, we interpret $\N_{\Zd}$ as the average of a random spectral measure. In the second part, we prove Lemma~\ref{l:ids}, which expresses $\N_{\Zd}$ as the limit of  normalised eigenvalue counting functions. The content of this section is  not used in the proof of our main result. 

\subsection{$H$ as a random operator: infinite graphs}\label{s:rand2}  For the case of infinite graphs,  the interpretation of $H[G]$ as a random operator requires the construction of Fourier transform on the infinite symmetric group, due to Vershik and Tsilevich \cite{VT}. The goal of the current section is to derive a formula for $\N_{\Zd}(\lambda)$ as a mixture of spectral functions $\N_{\Zd, \t}(\lambda)$:
\begin{equation}\label{eq:mixt}\N_{\Zd}(\lambda) =\int_{\Tab} d\Pl(\t) \, \N_{\Zd, \t}(\lambda) \end{equation}
(clarifications below), which  resembles the representation of the integrated density of states of a conventional random operator as a mixture of the random spectral measures corresponding to realisations of the randomness.

Similarly to the conventional setting, (\ref{eq:mixt}) reflects the fact that the operator $H[\Zd]$ itself is a direct integral of operators $\widehat{H[\Zd]}_{\t}$ corresponding to the slices $\t \in \Tab$ (however, we restrict the discussion to (\ref{eq:mixt}),  in order to avoid introducing further notation). Thus, (\ref{eq:mixt}) justifies, to some extent, the analogy between $\N_{\Zd}(\lambda)$  and the density of states of random operators: the integrand in the right-hand side of (\ref{eq:mixt}) is a spectral function of a random operator, with $\Tab$ playing the r\^ole of a probability space, and thus motivates both the results and the open questions that we have stated in the introduction.

To define the objects appearing in  the formula (\ref{eq:mixt}) (the domain of integration $\Tab$, the probability measure $\Pl$, and the spectral function $\N_{\widehat {H[\Zd]}; 0, \t}$), we need to recall (without proofs) the main definitions and results from \cite{VT}. The desired formula appears in Corollary~\ref{cor:vt} at the end of this section as an easy corollary of the general theory of \cite{VT}, allowing to block-diagonalise any operators commuting with the right action of the infinite symmetric group.

\paragraph{Young diagrams, tableaux and bitableaux, and the Plancherel measure} Denote by $\mathbb Y_N$ the collection of Young diagrams with $N$ boxes (i.e.\ collections of $N$ identical squares arranged in left-aligned rows of non-decreasing length; these are in on-to-one correspondence with partitions of $N$). The Young graph $\mathbb Y$ is the directed graph the vertices of which are the elements of $\bigcup_{N \geq 1} \mathbb Y_N$, with an edge $\lam \to \lam'$ for each pair $\lam \in \mathbb Y_N$, $\lam'\in \mathbb Y_{N+1}$ such that  $\lam \subset \lam'$. 

A Young tableau  of length $N$ is a path $\t = (\t_n)_{n=1}^N$ from $\t_1 = \square \in \mathbb Y_1$ to some $\t_N = \lam \in \mathbb Y_N$; the set of all Young tableaux of length $N$ is denoted $\Tab_N$. The dimension $\dim \lam$ of $\lam \in \mathbb Y_N$ is the number of tableaux terminating at $\lam$, An infinite Young tableau is an infinte path from $\square$; the set of infinite tableaux is denoted $\Tab$. For an infinite tableau $\t$, let $\head_N(\t) = (\t_1, \cdots, \t_N) \in \Tab_N$ and $\tail_N(\t) = (\t_{N}, \t_{N+1}, \cdots)$. The functions $\head_N$ induce a topology on $\Tab$ (the topology of projective limit). 

A bitableau of length $N$ is a pair  $(\s, \t) \in \Tab_N^2$ such that $\s_N = \t_N$. The collection of bitableaux of length $N$ is denoted $\Bitab_N$. For $\s, \t \in \Tab$, we write $\s \sim_N \t$ if $\tail_N(\s) = \tail_N(\t)$, and $\s \sim \t$ if $\s \sim_N \t$ for some $N\geq 1$. The equivalence class of $\t$ with respect to $\sim$ is denoted $\tail(\t)$. An infinite bitableau is a pair $(\s , \t) \in \Tab^2$ such that $\s \sim \t$. The space of bitableaux is denoted $\Bitab$. It is equipped with the topology of inductive limit induced by the functions 
\[ \{ \t \in \Tab \, | \, \t_N = \s_N \} \to \Bitab~, \quad \t \mapsto (\t, (\s_1, \cdots, \s_{N-1}, \tail_N(\t)))  \]
(indexed by $N \geq 1$ and $\s \in \Tab_N$).	

The Plancherel measure $\Pl_N$ on $\Tab_N$ is  the probability measure $\Pl_N$, $\Pl_N(\{\t_N\}) = \frac{\dim \t_N}{N!}$. The Plancherel measure $\Pl$ on $\Tab$ is defined by the relations $(\head_N)_* \Pl = \Pl_N$, $N = 1, 2, \cdots$.  It defines a Markov process on $\mathbb Y$, starting from $\square \in \mathbb Y_1$ and with transition probabilities 
\begin{equation}\label{eq:trans}\p(\lam \to \lam') = \frac{\dim \lam'}{(N+1) \dim \lam}~, \quad \lam \in \mathbb Y_N~, \, \lam' \in \mathbb Y_{N+1}~, \,\, \lam \subset \lam'~. \end{equation}
The full Plancherel measure $\tilde \Pl_N$ on $\Bitab_N$ is defined as
\[\begin{split} \tilde\Pl_N &= \int_{\Tab_N} d\Pl_N(\t) \sum_{\s \in \Tab_N\, : \, \s_N = \t_N} \delta_{(\t,\s)} \\
&= \sum_{\t \in \Tab_N}  \frac{\dim \t_N}{N!} \sum_{\s \in \Tab_N: \, \s_N = \t_N} \delta_{(\t,\s)} 
= \sum_{\lam \in \mathbb Y_N} \frac{\dim \lam}{N!} \sum_{\s,\t \in \Tab_N: \s_N = \t_N = \lam} \delta_{(\t,\s)}~. \end{split}\]
  The full Plancherel measure $\tilde \Pl$ on $\Bitab$ is defined as 
\[ \tilde\Pl = \int_{\Tab} d\Pl(\t) \sum_{\s \in \Tab \, : \, \s \sim \t} \delta_{(\t,\s)}~, \quad \text{i.e.} \quad
\int_{\Bitab} f(\s, \t) \, d\tilde\Pl(\s, \t) = \int_{\Tab} d\Pl(\t) \, \sum_{\t \sim \s} f(\s,\t)~.\]
Note that the full Plancherel measure is a $\sigma$-finite measure rather than a probability measure.

\paragraph{Fourier transform on the symmetric group} Let $A_1 \subset A_2 \subset A_2 \cdots$ be an ascending chain of sets, $|A_N| = N$, and let $\Sym_N = \Sym[A_N]$ be the corresponding symmetric groups; also let $A = \cup_N A_N$, $\Sym = \cup_N \Sym_N$. From classical representation theory, the irreducible representations of $\Sym_N$ are in one-to-one correspondence with $\mathbb Y_N$, i.e.\ $\operatorname{Irrep} \Sym_N \simeq \mathbb Y_N$. We denote the representation corresponding to $\lam \in \mathbb Y_N$ by the same letter $\lam$, and by $E_\lam$ -- the ambient space. The dimension of $E_\lam$ is equal to $\dim \lam$ defined above, and, moreover, $E_\lam$ has a special orthonormal basis $(h_\t)$ labelled by tableaux $\t \in \Tab_N$ terminating at $\lam$.  Note that the pushforward of $\Pl_N$ to $\mathbb Y_N$ coincides with the Plancherel measure defined in the Introduction.

The Fourier transform $\mathcal F_N: \ell_2(\Sym_N) \to L_2(\Bitab_N, \tilde\Pl_N)$ is defined via
\[ \left(\mathcal F_N \sum\limits_{\pi \in \Sym_N} c_\pi \pi \right)(\s, \t) = \begin{cases}
\sum_{\pi \in \Sym_N} c_\pi \langle \lam(\pi) h_\s,  h_\t \rangle~, & \s_N = \t_N = \lam \\
0~, &\text{otherwise}\end{cases}\]
(here we have taken the liberty to omit the adjoint from the definition in \cite{VT}). It is a unitary isomorphism of Hilbert spaces. The Fourier transform $\mathcal F: \ell_2(\Sym) \to L_2(\Bitab, \tilde\Pl)$ is first defined on $\ell_2(\Sym_N)$ via
\[ \left(\mathcal F u \right)(\s, \t) = \begin{cases}
\left(\mathcal F_N u \right)(\head_N (\s), \head_N(\t))~, &\s \sim_N \t\\
0~, \end{cases}\]
and then extended by continuity. It is a unitary isomorphism between $\ell_2(\Sym)$ and $L_2(\Bitab, \tilde\Pl)$.
Following \cite[6.6]{VT}, we remark that $(\mathcal F \mathbbm 1)(\s,\t) = \mathbbm 1_{\s=\t}$. We also remark that while $\mathcal F_N$ is equivalent to the natural decomposition of a function on the symmetric group in the basis of matrix elements of irreducible representations, $\mathcal F$ is not directly related to any of the classical notions of irreducible representations of the infinite symmetric group.

\paragraph{Spectral decomposition of multiplication operators} An element $X \in \mathbb C[\Sym]$ (and more generally of the von Neumann group algebra of $\Sym$) defines a bounded operator on $\ell_2(\Sym)$ acting by multiplication from the left. Then $\hat X = \mathcal F X \mathcal F^*$ acts via
\[ (\hat X f)(\s, \t) = \sum_{\r \sim \s} \hat X(\s, \r) f(\r, \t)~.\]
Thus $\hat X$ induces a family of operators $\hat X_\t: \ell_2(\tail(\t)) \to \ell_2(\tail(\t))$, $\t \in \Tab$, given by
\[ \hat X_\t f(\s) = \sum_{\r \sim \s} \hat X(\s, \r) f(\r)\]
(i.e.\ $\hat X$ is the direct integral of $\hat X_t$ with respect to $d\Pl(t)$). 
If $X$ is self-adjoint, denote by $\N_{X; \mathbbm 1}(\lambda)$ the spectral function of $X$ corresponding to the vector $\mathbbm 1 \in \ell_2(\Sym)$ (the identity permutation), i.e.\ a cumulative distribution function defined by the property
\[ \forall  p \in \mathbb C[\lambda] \quad \int_{\mathbb R} p(\lambda) \, d\N_{X; \mathbbm 1}(\lambda) = \langle p(X) \,  \mathbbm 1, \mathbbm 1\rangle~, \]
and by $\N_{\hat X; \t}(\lambda)$ -- the similarly defined spectral function of $\hat X_\t$ corresponding to $\mathbbm 1_\t$:
\[ \forall  p \in \mathbb C[\lambda] \quad \int_{\mathbb R} p(\lambda) \, d\N_{\hat X; \t}(\lambda) = \langle p(\hat X_\t) \,  \mathbbm 1_\t, \mathbbm 1_\t\rangle = p(\hat X)(\t, \t)~. \]
 Then we have:
\begin{equation}\label{eq:decomp-prelim}  \N_{X; \mathbbm 1}(\lambda) = \int_{\Tab} d\Pl(\t) \N_{\hat X; \t}(\lambda)~.\end{equation}
Indeed, for any polynomial $p \in\mathbb C[\lambda]$,
\[ \begin{split}
\int_{\mathbb R} p(\lambda) \, d\N_{X; \mathbbm 1}(\lambda)
&= \langle p(X) \,  \mathbbm 1, \mathbbm 1\rangle = \langle p(\hat X) \,  \mathcal F \mathbbm 1, \mathcal F \mathbbm 1\rangle \\
&= \int d\tilde\Pl(s,t) \big( p(\hat X) \,  \mathcal F \mathbbm 1\big)(\s, \t) \overline{ (\mathcal F \mathbbm 1)(\s,\t)} \\
&= \int d\tilde\Pl(s,t) \sum_{\r \sim \s} p(\hat X)(\s, \r) \, \mathbbm 1_{\r = \t}\mathbbm 1_{\s = \t}  
=\int d\tilde\Pl(s,t) p(\hat X)(\s, \s) \mathbbm 1_{\s = \t} \\
&= \int d\Pl(\s) p(\hat X)(\s, \s) = \int_{\Tab} d\Pl(\s) \int_{\mathbb R} p(\lambda)\, d\N_{\hat X; \s}(\lambda)~.
\end{split}\]
In probabilistic terms, (\ref{eq:decomp-prelim}) expresses $\N_{X; \mathbbm 1}(\lambda) $ as the expectation of the random measure $\N_{\hat X; \t}(\lambda)$ when $\t$ is picked at random according to $\Pl$. This is almost what we need, except that the operator $H[G]$ which is studied in the current paper is not a multiplication operator such as $X$, but rather a block matrix formed from such operators.

\paragraph{Block multiplication operators} The construction from the previous paragraph, and particularly the relation (\ref{eq:decomp-prelim}), is raised to $\ell_2(B\to \ell_2(\Sym))$, where $B$ is an arbitrary set, as follows. Let $X$ be an operator acting  on this space via
\[ (X \psi)(x) = \sum_{y \in B} X(x, y) \psi(y)~,  \]
where, $X(x,y) \in \mathbb C[\Sym]$, and, say,
\[ \sup_{x \in B} \,\,\sum_{y \in B} (\|X(x, y) \| + \|X(y,x)\|) < \infty \]
(to ensure that the operator is bounded). Then $\hat X_\t$ acts on $\ell_2(B\times \tail(\t))$ via
\[ (\hat X f)(x,\s) = \sum_{y \in B; \r \sim \s} \widehat{ X(x, y) }(\s, \r) f(y, \r)~. \]
If  $X$ is self-adjoint, i.e.\ $X(y, x) = X(x, y)^*$ for any $x,y \in B$, we can define the spectral function $\N_{X; x, \mathbbm 1}$ of $X$ at $\delta_x \mathbbm 1$, and also the spectral function $\N_{\hat X; x, \t}$ of $\hat X_\t$ at $(x, \t)$. Then we have:
\[  \N_{X; x, \mathbbm 1}(\lambda) = \int_{\Tab} d\Pl(\t) \, \N_{\hat X; x, \t}(\lambda)~.\]

Now we can specialise to our operator $H[\Zd]$  by fixing an arbitrary enumeration of $\Zd$ and letting $B = \Zd$. We thus obtain a representation of the integrated density of states as the expectation of a random spectral measure, which we promised in (\ref{eq:mixt}):
\begin{cor}\label{cor:vt} Fix an ordering of $\Zd$ (i.e.\ an ascending chain $A_1 \subset A_2 \subset \cdots \subset A_N\subset \cdots$ such that $|A_N| = N$ and $\cup_N A_N = \Zd$). The integrated density of states $\N_{\Zd}(\lambda) = \N_{H[\Zd]; 0, \mathbbm 1}(\lambda)$ admits the decomposition
\[\N_{\Zd}(\lambda) =  \int_{\Tab} d\Pl(\t) \N_{\Zd, \t}(\lambda)~, \quad \N_{\Zd, \t}(\lambda) = \N_{\widehat {H[\Zd]}; 0, \t}(\lambda)~.\]
\end{cor}

\subsection{Proof of Lemma~\ref{l:ids}}\label{s:ids}

To prove the convergence of $\N_{B_{L,d}}$ to $\N_{\Zd}$, it suffices to show that for any $n \geq 1$
\begin{equation}\label{eq:l:ids:need}\lim_{L \to \infty} \frac{1}{|B_{L,d}| \times |B_{L,d}|!} \operatorname{tr} H[B_{L,d}]^n = \tau\left( H[\Zd]^n(0,0)\right)~.\end{equation}
The expression under the limit equals 
\[ \frac{1}{|B_{L,d}|} \sum_{x \in B_{L,d}} \tau \left( H[B_{L,d}]^n(x, x)\right)~. \]
All the $|B_{L,d}| = (2L+1)^d$ terms in the sum are bounded by $C_n = (4d)^n$. Most of the terms, namely, those that correspond to $x$ at distance $> n/2$ from the boundary of $B_{L,d}$, are equal to the right-hand side of (\ref{eq:l:ids:need}), and this concludes the proof of convergence, and we turn to the proof of the equality between the support of $d\N_{B_{L,d}}$ and $[-2d, 2d]$.

The operator   of multiplication by a transposition is unitary, hence $\| H[B_{L,d}]  \| \leq 2d$ for any $L$, thus definitely the support of the integrated density of states is contained in $[-2d, 2d]$.  To see that equality holds, let $\lambda \in [-2d, 2d]$; choose  $\alpha \in  (\mathbb R / \mathbb Z)^d$ such that $2 \sum_{j=1}^d \cos (2 \pi \alpha_j) = \lambda$. Let $\epsilon > 0$, and choose $\ell =  \lceil C_* \epsilon^{-2} \rceil$, where  $C_*$ is large enough. Then for any $B \supset B_{\ell,d}$ and any function $\psi \in \mathfrak H[B]$ of the form
\begin{equation}\label{eq:tmptmp1} \psi(x)  = \begin{cases}
e^{2\pi i \langle \alpha, x \rangle} \left[ \sum\limits_{\pi \in \Sym[B]}  c_\pi \pi \right]~, & \|x \|\leq \lceil \frac \ell 2 \rceil \\
0
\end{cases} \end{equation}
where $c_\bullet$ is constant on the right cosets of $\Sym[B_{\ell,d}]$, we have:
\begin{equation}\label{eq:approxev} \|H[B] \psi - \lambda \psi \| \leq C \ell^{-1/2} \|\psi\| \leq \epsilon \|\psi\|~,\end{equation}
since by construction the expression in the square brackets of (\ref{eq:tmptmp1}) is invariant under all $(x \, y)$ such that $\| x \| \leq \lceil \frac{\ell}2 \rceil$ and $y$ is adjacent to $x$, and therefore the computation is reduced to applying the identity matrix of $\Zd$ to the function equal to $\exp(2\pi i \langle \alpha, x \rangle)$ for $x \in B_{\frac\ell2, d}$ and to zero outside this box.
Observe that the space of such functions $\psi$ is of dimension $|B|!/|B_{\ell,d}|!$

The box $B_{\ell,d}$ can be replaced with  any congruent box  lying inside $B$; moreover, any linear combination of functions of this form coming from disjoint congruent boxes still satisfies (\ref{eq:approxev}). Consequently, the operator $H[B_{k\ell,d}]$ has at least
\[ k^d \times (2k\ell+1)^d! /  (2\ell+1)^d ! \]
eigenvalues in $(\lambda - \epsilon, \lambda + \epsilon)$, counting multiplicity. Letting $k \to \infty$, we obtain:
\[ \N_{\Zd}(\lambda+ \epsilon) - \N_{\Zd}(\lambda -  \epsilon) \geq
\lim_{k \to \infty} \frac{ k^d \times (2k\ell+1)^d! /  (2\ell+1)^d !}{(2k\ell + 1)^d \times (2k\ell + 1)^d!} > 0~.
\]
This holds for any $\epsilon > 0$, hence $\lambda$ lies in the support.
\qed

\section{Proof of the main result}

\subsection{Estimate on the moments of $H$}

 We first deduce Theorem~\ref{thm} from Theorem~\ref{thm:prob}; the proof of the latter is spread over the next two subsections.

\begin{proof}[Proof of Theorem~\ref{thm}] First, rewrite Theorem~\ref{thm:prob} as an estimate on  the moments of $H = H[\Zd]$: for each $d \geq 2$, there exist $C > 0$ and $c> 0$ such that for any $n \geq 0$
\begin{equation}\label{eq:propmain} c \exp\left\{-C n^{\frac{d}{d+2}} \log^{\frac{2}{d+2}}(n+2)\right\} \leq \frac{\tau( H^{2n}(0,0))}{(2d)^{2n}} \leq C \exp\left\{-c n^{\frac{d}{d+2}} \right\}~.\end{equation}

\noindent {\em Upper bound}:
According to   the second inequality in (\ref{eq:propmain}),
\[\begin{split} C e^{-cn^{\frac{d}{d+2}}} 
&\geq \frac{\tau(H^{2n}(0,0))}{(2d)^{2n}}\\
&= \int_{-2d}^{2d} (\lambda/(2d)) ^{2n} d\N(\lambda) 
\geq \int_{-2d}^{-2d+\epsilon} (\lambda/(2d)) ^{2n} d\N(\lambda)  \geq (1 - \frac{\epsilon}{2d})^{2n} \N(-2d+\epsilon)
\end{split}\]
Take $n = \lceil  a \epsilon^{-\frac{d+2}{2}} \rceil$, then 
\[ (1 - \frac{\epsilon}{2d})^{2n} \geq \exp(- C a \epsilon^{-\frac{d}2})~,\quad
e^{-cn^{\frac{d}{d+2}}}  \leq \exp(-c a^{\frac{d}{d+2}} \epsilon^{\frac{d}{2}})~, \]
whence
\[ \N(-2d+\epsilon) \leq C \exp( - \left[ c a^{\frac{d}{d+2}} - Ca\right]\epsilon^{\frac{d}{2}})~,\]
and it remains to take $a > 0$ small enough to ensure that the expression in the square brackets is positive.

\medskip\noindent
{\em Lower bound}: 
According to the the first inequality in (\ref{eq:propmain}), 
\[\begin{split}
c \exp\left\{-C n^{\frac{d}{d+2}} \log^{\frac{2}{d+2}}(n+2)\right\} 
&\leq 2 \int_{-2d}^0  (\lambda/(2d)) ^{2n} d\N(\lambda) \\
&\leq 2 \left\{ \N(-2d+\epsilon) + \exp(-\frac{\epsilon n}{d})\right\}
\end{split}\]
whence
\[  \N(-2d+\epsilon) \geq \frac{c}{2} \exp\left\{-C n^{\frac{d}{d+2}} \log^{\frac{2}{d+2}} (n+2)\right\}  - \exp(-\frac{\epsilon n}{d})~. \]
Now we take $n = \lceil a  \epsilon^{-\frac{d+2}{2}} \log (\frac{1}\epsilon + 2)\rceil$, where $a>0$ is sufficiently small.
\end{proof}

\subsection{Some properties of the random walk on $\Zd$}
Let $\overrightarrow{x} = (x_j)_{j \geq 0}$ be a path on $\Zd$ starting from $x_0 = 0$.  For $y \in \Zd$, let 
\[ j^*(y) = j^*(y; \overrightarrow{x}) = \min \left\{ j \geq 0 \, | \, x_{j} = y \right\}  \in \mathbb Z_+ \cup \{+\infty\}\]
be the time of the first visit to $y$.  Let
\[ R_n[\overrightarrow x] = \{ y_j \, | \, 0 \leq j < n \} = \{ y \, | \, j^*(y; \overrightarrow x ) < n \} \]
be the range of the head of $\overrightarrow x$, and let $R[\overrightarrow x] = \cup_{n \geq 0} R_n[\overrightarrow x]$.  A vertex $y \in \Zd$ is called $\overrightarrow x$-{\em flexible} if $x \in R[\overrightarrow x]$ and 
\[ x_{j^*(y)+2} - x_{j^*(y)+1} \neq \pm (x_{j^*(y)+1} - x_{j^*(y)})~; \]
we denote the collection of all flexible vertices by $F[\overrightarrow x]$. Set 
\[ R_n^{\operatorname{even}}[\overrightarrow x] = R_n[\overrightarrow x] \cap \Even~, \quad F^{\operatorname{even}}[\overrightarrow x] = F[\overrightarrow x] \cap \Even~,\]
where $\Even \subset \Zd$ consists of the vertices with even sum of coordinates. Finally, denote by $y_+^*(j) = y_+^*(j; \overrightarrow x)$ the $j$-th visited even vertex, i.e.\ the vertex $y \in \Even$ such that for some $n \geq 0$
\[ R_{2n}^{\operatorname{even}}[\overrightarrow x] = R_{2n-2}^{\operatorname{even}}[\overrightarrow x] \uplus  \{y\}~, \quad  |R_{n+2}[\overrightarrow x]|=j \]
(where formally $R_{-2} = \varnothing$). These definitions are illustrated by Figure~\ref{fig:flex}.
\begin{figure}[h]
\begin{center}
\includegraphics[scale=2]{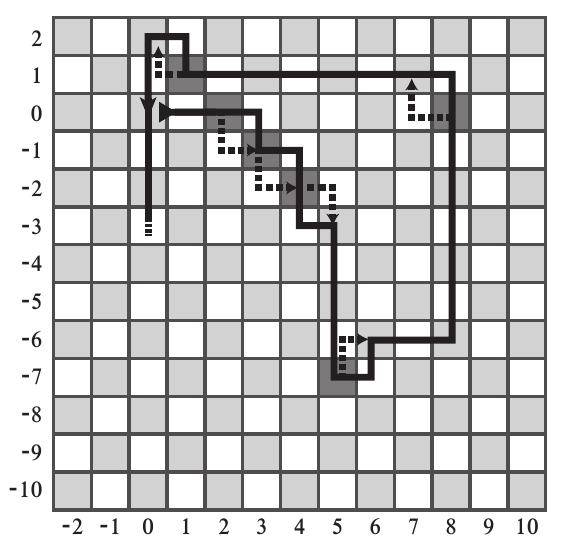}
\end{center}
\caption{The 21-edges long initial segment of an infinite path on $\mathbb Z^2$ (solid black). The even flexible vertices are highlighted by a darker shade of grey. The equivalence class of this path in the sense of (\ref{eq:equiv})  below consists of the $64$ paths obtained by replacing some of the two-step segments starting with an even flexible vertex with their dashed counterparts.}\label{fig:flex}
\end{figure}

We apply these definitions to a realisation $\overrightarrow X = (X_j)_{j \geq 0}$ of the simple random walk on $\Zd$, starting from $X_0 = 0$.
The main result of this section is
\begin{prop}\label{prop:flex}
 For each $d \geq 2$, there exist $a > 0$ and $A > 0$ such that for any $n \geq 1$
\[ \mathbb P \left\{ | R_n^{\operatorname{even}}[\overrightarrow X] \cap F^{\operatorname{even}}[\overrightarrow X] | \leq a n^{\frac{d}{d+2}} \right\} \leq A \exp\{-n^{\frac{d}{d+2}}\}~. \] 
\end{prop}

\noindent The proof of Proposition~\ref{prop:flex} rests on two lemmata. The first one, Lemma~\ref{l:1}, ensures that the range of the random walk is not too small. It is closely related to the large deviation estimate of  Donsker and Varadhan \cite{DV} on the range of the random walk. Since we only need a crude version of a special case of the result of \cite{DV}, we provide a self-contained proof. The second one, Lemma~\ref{l:2} ensures that a sizeable part of the vertices in the range are flexible. 

\begin{lemma}\label{l:1}
For each $d \geq 1$ there exist $a_1 > 0$ and $A_1 > 0$ such that for any $n \geq 1$
\[ \mathbb P \left\{ | R_n^{\operatorname{even}}[\overrightarrow X]| \leq a_1  n^{\frac{d}{d+2}} \right\} \leq A_1 \exp\{-n^{\frac{d}{d+2}}\}~. \] 
\end{lemma}

\begin{lemma}\label{l:2}
For each $d \geq 2$ and $b > 0$ there exist $a_2 > 0$ and $A_2 > 0$ such that for any $m \geq 1$
\[ \mathbb P \left\{  \left|  \left\{ 0 \leq j < m \, | \, y_+^*(j, \overrightarrow X) \in F^{\operatorname{even}}[\overrightarrow X] \right\} \right| \leq a_2   m \right\} \leq A_2 e^{-b m}~, \]
 \end{lemma}

\begin{proof}[Proof of Proposition~\ref{prop:flex}]
Apply Lemma~\ref{l:1} and then Lemma~\ref{l:2} with $b = 1/a_1$ and $m = a_1 n^{\frac{d}{d+2}}$; set $a = a_1 a_2$ and $A = A_1 + A_2$. \end{proof}

\paragraph{On the spectrum of the Dirichlet Laplacian}
Here we prove Lemma~\ref{l:1.5}, which is used in the proof of Lemma~\ref{l:1}, and Lemma~\ref{l:3},   required for the proof of the lower bound in Theorem~\ref{thm:prob}. Both facts are well known (the first one is a discrete Sobolev inequality, the second one is an explicit computation).

 Denote by $P$ the generator of the random walk $\overrightarrow X$,
\[ (Pu)(x) = \frac1{2d} \sum_{y \sim x} u(y)~. \]
Denote by $\mathbbm 1_R$ the indicator of a set $R \subset \mathbb Z^2$, and also the corresponding multiplication operator on $\ell_2(\Zd)$.

\begin{lemma}\label{l:1.5} Let $d \geq 1$, and let $R \subset \Zd$ be a finite set. Then
\[ \| \mathbbm 1_R P \mathbbm 1_R \| \leq 1 - \frac{c_1}{\operatorname{diam}^2 R} \leq 1 - \frac{c_2}{|R|^{\frac2d}}~,\]
where $c_1,c_2>0$ may depend on $d$ but not on $R$.
\end{lemma}
 
\begin{lemma}\label{l:3} Let $d \geq 1$ and let $B_L = [-L, L]^d$. Then  $\mathbbm 1_{B_L} P \mathbbm 1_R$ has an eigenvalue $\lambda \geq 1 - c_3 L^{-2}$ such that the corresponding eigenfunction $u$, $\| u \| = 1$, satisfies $|u(0)| \geq c_4 L^{-d/2}$, where $c_3,c_4>0$ may depend on $d$ but not on $L$.
\end{lemma}

\begin{proof}[Proof of Lemma~\ref{l:1.5}] Let $u\in\ell_2(\Zd)$ be a function supported in $R$. We need to show that 
\begin{equation}\label{eq:need-l1.5}  -(1 - c_1 / \operatorname{diam}^2 R) \|u\|^2 \leq \langle \mathbbm 1_R P \mathbbm 1_R u, u \rangle \leq (1 - c_1 / \operatorname{diam}^2 R) \|u\|^2~. \end{equation}
Observe that
\begin{equation}\label{eq:l-1.5-1} \begin{split}
\langle \mathbbm 1_R P \mathbbm 1_R u, u \rangle
&= \langle P u, u \rangle \\
&= \frac{1}{2d} \sum_{x \in \Zd} \sum_{y \sim x} u(x) \overline{u(y)} \\
&= \frac{1}{2d} \sum_{x \in \Zd} \sum_{i=1}^d u(x) (\overline{u(x + e_i)} +  \overline{u(x - e_i)})~. 
\end{split}\end{equation}
Let us prove, for example, the second inequality in (\ref{eq:need-l1.5}).
For each $x' \in \mathbb Z^{d-1}$, 
\[\begin{split} 
&\frac{1}{2d}  \sum_{x'' \in \mathbb Z} u(x'', x') (\overline{u(x''+1, x')} + \overline{u(x''-1,x')}) \\
&\qquad= \frac{1}{d} \sum_{x''} |u(x'', x')|^2 - \frac1{2d} \sum_{x''} |u(x'', x') - u(x''+1,x')|^2 \\
&\qquad\leq \frac{1}{d} \sum_{x''} |u(x'', x')|^2 (1 - \frac{c_1}{\operatorname{diam}^2 R})\end{split}\]
(for example, since $\max_{x''} |u(x'', x')|^2 \geq \sum_{x''} |u(x'', x')|^2 / \operatorname{diam} R $). A similar inequality holds  for the other $d-1$ terms in (\ref{eq:l-1.5-1}). 
\end{proof}

\begin{proof}[Proof of Lemma~\ref{l:3}] Take  
\[ u(x) = \begin{cases} c_L \prod_{i=1}^d \cos \frac{\pi x(i)}{2(L+1)}~, & \|x\| \leq L \\
0~,\end{cases}\] 
where $x(i)$ is the $i$-th co\"ordinate of $x$, and $c_L$ is chosen  so that $\|u \|=1$. Then $|u(0)| = |c_L| \geq c L^{-d/2}$, and 
\[ P u = \left[ \frac{1}{d} \sum_{j=1}^d \cos \frac{\pi}{2L} \right]  u~, \quad  \frac{1}{d} \sum_{j=1}^d \cos \frac{\pi}{2L} \geq 1 - \frac{\pi^2}{8L^2}~.\qedhere \]
\end{proof}

\paragraph{Proof of Lemma~\ref{l:1}}
Let $R^{\operatorname{even}} \subset \Even$ be a set of cardinality $|R^{\operatorname{even}}| \leq a_1 n^{\frac{d}{d+2}}$. If $R_n^{\operatorname{even}}[\overrightarrow X] \subset R^{\operatorname{even}}$, then 
\[ R_n[\overrightarrow X] \subset R  {=} \{ x \in \mathbb Z^2 \, \mid \, \operatorname{dist}(x,R^{\operatorname{even}}) \leq 1\}~; \quad
|R| \leq (2d+1) |R^{\operatorname{even}}| \leq (2d+1) a_1 n^{\frac{d}{d+2}}~.\]
According to Lemma~\ref{l:1.5},
\[\begin{split} \mathbb P \{ X_0, \cdots, X_{n} \in R \} 
&= \sum_{x \in R} (\mathbbm 1_{R} P \mathbbm 1_{R})^{n}(x, 0) \\
&\leq |R| \| \mathbbm 1_{R} P \mathbbm 1_{R} \|^{n} \\
&\leq |R| (1 - \frac{c_2}{|R|^{\frac{2}{d}}})^{n} \leq (2d+1) a_1 n^{\frac{d}{d+2}}  \exp(- \frac{c_3}{a_1} n^{\frac{d}{d+2}})~.
\end{split}\]
On the other hand, $R_n[\overrightarrow X]$ is connected and contains the origin, hence the number of ways to choose $R$ is at most $C^{(2d+1)a_1 n^{\frac{d}{d+2}}} = \exp\{C_1 a_1 n^{\frac{d}{d+2}}\}$. Therefore
\begin{equation}\label{eq:in-l-1}\mathbb P\left\{ |R_n^{\operatorname{even}}[\overrightarrow X]| \leq a_1 \sqrt n \right\} 
\leq (2d+1) a_1 n^{\frac{d}{d+2}}  \exp(-\frac{c_3}{a_1} n^{\frac{d}{d+2}} + C_1 a_1 n^{\frac{d}{d+2}})~. \end{equation}
For sufficiently small $a_1$, $\frac{c_3}{a_1} - C_1 a_1 \geq 2$, and then
\[ \text{RHS of } (\ref{eq:in-l-1}) \leq (2d+1) a_1 n^{\frac{d}{d+2}}  \exp(- 2 n^{\frac{d}{d+2}}) \leq A_1 \exp(- n^{\frac{d}{d+2}})~. \qed \]

\paragraph{Proof of Lemma~\ref{l:2}} Define a new random walk $\overrightarrow Z$, as follows.  Let $(W(j,k,\iota))_{j,k\geq 0, \iota\in \{1,2\}}$ be an array of independent, identically distributed random variables taking on each of the values in $\{\pm e_i, 1 \leq i \leq d\}$ with probability $1/(2d)$.
Then set $Z_0 = 0$ and inductively
\[\begin{split}
Z_{2n+1} &= Z_{2n}+ W(j^*(Z_n), \# \{ 0 \leq n' < n \, | \, Z_{n'} = Z_n \}, 1)~; \\
Z_{2n+2} &= Z_{2n+1}+ W(j^*(Z_n), \# \{ 0 \leq n' < n \, | \, Z_{n'} = Z_n \}, 2)~. 
\end{split}\]
Then $\overrightarrow Z \overset{\text{distr}}{=} \overrightarrow X$, and hence
\[ \begin{split}
&\mathbb P \left\{ | \{ 0 \leq j < m \, \big | \, y^*(j; \overrightarrow X) \in F^{\operatorname{even}}[\overrightarrow X]\}| \leq a_2 m  \right\} \\
&\quad= \mathbb P\left\{ | \{ 0 \leq j < m \, \big | \, W(j,0,2) \neq \pm W(j,0,1)\}| \leq a_2 m  \right\} 
\leq A_2 e^{-b m} \end{split}\]
provided that $a_2$ is chosen to be small enough. \qed

\subsection{Proof of Theorem~\ref{thm:prob}}
Let $\overrightarrow X$ be the simple random walk starting from the origin. Then
\[ \frac{1}{(2d)^{2n}}\tau(H^{2n}(0, 0)) = \mathbb P \left\{ \pi_{2n}[\overrightarrow X] \overset{\text{def}}{=} (X_0 \, X_1) (X_1 \, X_2) \cdots (X_{2n-1} X_{2n}) = \mathbbm 1 \right\}~. \]

\paragraph{Upper bound} 
Consider the collection 
\[ 
\Omega_{2n} = \{ \overrightarrow x \, \mid \, | R_{2n}^{\operatorname{even}}[\overrightarrow x] \cap F^{\operatorname{even}}[\overrightarrow x] | \geq a (2n)^{\frac{d}{d+2}}\}
\]
of paths starting from the origin, where $a$ is as in Proposition~\ref{prop:flex}, and let $\Omega_{2n}'$ be the factor of $\Omega_{2n}$ by the heads $(x_0, x_1, \cdots, x_{2n})$ (i.e.\ by the relation $\sim_{2n}$, where $\overrightarrow x \sim_{2n} \overrightarrow x'$ if $x_j = x'_j$ for all $j \leq  2n$). Then (by the cited proposition)
\[ \mathbb P(\overrightarrow X \in \Omega_{2n}) \geq 1 - A \exp(-(2n)^{\frac{d}{d+2}})~. \] 
We call $\overrightarrow x, \overrightarrow x' \in \Omega_{2n}$  equivalent ($\overrightarrow x \sim \overrightarrow x'$) if for any $1 \leq j \leq 2n$ 
\begin{equation}\label{eq:equiv} \begin{cases}
\text{either $x_j = x_j'$}\\
\text{or $j$ is odd, $x_{j-1} \in F[\overrightarrow x]$ and $j-1 = j^*(x_{j-1}; \overrightarrow x)$.}
\end{cases}\end{equation}
 In other words, if $j-1$ is the time of the first visit of $\overrightarrow x$ to a flexible even vertex $x$, we are allowed to replace $x_{j}$ with $x'_{j}=  x_{j-1} -x_{j} + x_{j+1}$. See Figure~\ref{fig:flex} for an illustration. Note that  $\sim$ is indeed an equivalence relation on $\Omega_{2n}$ and on $\Omega_{2n}'$. By the construction of $\Omega_{2n}$, every $\sim$-equivalence class $\mathfrak x$ on $\Omega_{2n}'$ is of cardinality 
\[ |\mathfrak x| \geq 2^{a (2n)^{\frac{d}{d+2}}}~. \]
{\em Claim:} in each class $\mathfrak x$ on $\Omega_{2n}'$, there is at most one representitative $[\overrightarrow x]$ for which $\pi_{2n}[\overrightarrow x] = \mathbbm 1$. Indeed, if $\overrightarrow x \sim \overrightarrow x'$ and $x_j \neq x_{j'}$ for some $j \leq 2n$, then 
\[ j-1 = j^*(x_j, \overrightarrow x) = j^*(x_j, \overrightarrow x')~. \]
Among such indices $j$, choose the smallest one, so that $x'_i = x_i$ for $i < j$. If $\pi_{2n}[\overrightarrow x] = \mathbbm 1$, then 
\[ \big( (x_{j} \, x_{j-1}) (x_{j-1} \, x_{j-2})   \cdots (x_{1} \, x_{0}) \big) x_{j-1} = x_{j-1}~, \]
whence 
\begin{equation}\label{eq:cond-x} \big( (x_{j-1} \, x_{j-2})   \cdots (x_{1} \, x_{0}) \big) x_{j-1} = x_{j}~.\end{equation}
Similarly, if $\pi_{2n}[\overrightarrow x'] = \mathbbm 1$, then 
 \begin{equation}\label{eq:cond-x'} \big( (x'_{j-1} \, x'_{j-2})   \cdots (x'_{1} \, x'_{0}) \big) x'_{j-1}  = \big( (x_{j-1} \, x_{j-2})   \cdots (x_{1} \, x_{0}) \big) x_{j-1} = x_{j}'~.\end{equation}
Clearly, at most one of the identities (\ref{eq:cond-x}) and (\ref{eq:cond-x'}) can hold. This proves the claim. 

Now the proof of the upper bound is concluded as follows:
\[\begin{split}
\mathbb P \left\{ \pi_{2n}[\overrightarrow X] = \mathbbm 1\right\}
&\leq \mathbb P \left\{ \overrightarrow X \notin \Omega_{2n}\right\} + \sum_{\mathfrak x} 2^{-a (2n)^{\frac{d}{d+2}}} \mathbb P \left\{ \overrightarrow X \in \mathfrak x\right\}  \\
&\leq A \exp(-(2n)^{\frac{d}{d+2}}) + 2^{-a (2n)^{\frac{d}{d+2}}} \leq C \exp(-c n^{-\frac{d}{d+2}})~,
\end{split}\]
as claimed.

\paragraph{Lower bound}
Consider the random walk $\overrightarrow X$ starting from the origin. According to Lemma~\ref{l:3}, we have for $1 \leq L\leq n^{0.49}$:
\begin{equation}\label{eq:forlowerbd}\begin{split}
 \mathbb P \left\{ R_{2n}[\overrightarrow X] \subset B_{L,d}~, \,\, X_{2n} = 0 \right\} &= (\mathbbm 1_{B_{L,d}} P \mathbbm 1_{B_{L,d}})^{2n}(0, 0) \\
&\geq c_4^2 L^{-d} (1-L^{-2})^{2n} \geq c' \exp(- C' L^{-2} n)~. 
\end{split}\end{equation}
Consider the decorated Markov process $\overrightarrow {\mathfrak X}_L = (X_j, \sigma_j)_{j \geq 0}$, $\sigma_{j+1} =  \sigma_j (X_{j+1} \, X_j)$ (with arbitrary initial points $(X_0, \sigma_0)$), and let 
\[ \mathfrak P_{n}((x, \sigma) \to (x', \sigma')) =  \mathbb P \left\{ X_1, \cdots, X_{n-1} \in B_L~, \, X_{n} = x',  \pi_n[\overrightarrow X] = \sigma'  \, \big| \, X_0 = x, \sigma_0 = \sigma \right\}~. \]
Then by (\ref{eq:forlowerbd}) and the Cauchy--Schwarz inequality
\[\begin{split} c' \exp(- C' L^{-2} n) 
&\leq \sum_{\sigma \in \Sym(B_{L,d})}  \mathfrak P_{2n}((0, \mathbbm 1) \to (0, \sigma)) \\
&= \sum_{\sigma, \sigma' \in \Sym(B_{L,d}), x \in B_L}  \mathfrak P_{n}((0, \mathbbm 1) \to (x, \sigma'))  \mathfrak P_{n}((0, \sigma) \to (x, \sigma') ) \\
&= \sum_{\sigma, \sigma' \in \Sym(B_{L,d}), x \in B_L}  \mathfrak P_{n}((0, \mathbbm 1) \to (x, \sigma'))  \mathfrak P_{n}((0, \mathbbm 1) \to (x, \sigma^{-1}\sigma') ) \\
&\leq |B_L|! \sum_{\sigma' \in \Sym(B_{L,d}), x \in B_L}  \mathfrak P_{n}((0, \mathbbm 1) \to (x, \sigma'))^2 
= |B_L|!\mathfrak P_{2n}((0, \mathbbm 1) \to (0, \mathbbm 1))~, 
\end{split}\]
whence
\[\begin{split} \mathbb P \left\{ \pi_{2n}[\overrightarrow X] = \mathbbm 1\right\} 
&\geq \mathfrak P_{2n}((0, \mathbbm 1) \to (0, \mathbbm 1)) \\
&\geq \frac{1}{|B_L|!} \times c' \exp(- C' L^{-2}n) \geq c' \exp(- C' L^{-2}n - (2L+1)^d \log (2L+1))  ~. 
\end{split}\]
Taking $L =n^{\frac{1}{d+2}} \log^{-\frac{1}{d+2}} (n+2)$, we obtain
\[\mathbb P \left\{ \pi_{2n}[\overrightarrow X] = \mathbbm 1\right\}  \geq c'' \exp(- C'' n^{\frac{d}{d+2}} \log^{\frac{2}{d+2}} (n+2))~,\]
as claimed. \qed

\paragraph{Acknowledgement.} We are grateful to Alexey Bufetov, Vadim Gorin, Matthias T\"aufer and Vlad Vysotsky for helpful discussions, and to Grigori Olshanski and Anatoly Vershik for the comments on the first draft of this work.

\end{document}